\documentclass[lettersize,journal]{IEEEtran}

\usepackage{amssymb}
\usepackage{amsmath}
\usepackage{graphicx}
\usepackage{cite}
\usepackage{dsfont}

\usepackage[binary-units]{siunitx}
\usepackage[caption=false,font=footnotesize]{subfig}
\usepackage{tabularx}
\usepackage{multirow}
\usepackage{diagbox}
\usepackage{xcolor}
\newtheorem{proposition}{Proposition}
\newtheorem{remark}{Remark}

\begin{document}
\title{Distributed  vs. Centralized Precoding in Cell-Free Systems: Impact of Realistic Per-AP Power Limits}   

\author{Wei~Jiang,~\IEEEmembership{Senior~Member,~IEEE,}~and~
        Hans~Dieter~Schotten
\thanks{W. Jiang and H. D. Schotten are with German Research Center for Artificial Intelligence (DFKI),  67663 Kaiserslautern, Germany, and are also with the department of Electrical and Computer Engineering, RPTU University Kaiserslautern-Landau, Germany. (e-mail: wei.jiang@dfki.de, schotten@eit.uni-kl.de).}
}

\maketitle

\begin{abstract}  
   In cell-free massive MIMO, centralized precoding is {theoretically known} to {remarkably} outperform its distributed counterparts, albeit {with} high implementation complexity. However, this letter highlights a practical limitation {often overlooked:} {widely used closed-form} centralized {precoders} are typically derived under a sum-power constraint, which often demands unrealistic power allocation that exceeds hardware capabilities. {When two simple heuristics (global power scaling and local normalization) are applied to enforce the per-AP instantaneous power constraint}, the centralized performance superiority disappears, making distributed precoding {a robust option}.
\end{abstract}

\IEEEpeerreviewmaketitle

\section{Introduction}

\IEEEPARstart{C}{ell-free} (CF) massive multi-input multi-output (CF-mMIMO) is a promising architecture owing to its provision of uniform service quality. { From an information-theoretic standpoint, centralized precoding leveraging system-wide channel-state information (CSI) is inherently superior to distributed approaches, as the set of feasible distributed precoders is a subset of centralized ones. }  

In \cite{Ref_nayebi2017precoding, du2021cell}, it is shown that centralized zero-forcing (ZF) precoding significantly outperforms distributed conjugate beamforming (CBF). 
Then, \cite{Ref_bjornson2020scalable} suggests that a centralized precoder using partial minimum mean square error (P-MMSE) surpasses locally implemented P-MMSE (LP-MMSE) in scalable CF systems with dynamic cooperation clustering.  
Even in decentralized frameworks, a theoretically optimal precoder called team MMSE is outperformed by its centralized equivalent \cite{miretti2021team}. 
The work \cite{Ref_atzeni2021distributed} shows that cooperative precoding can only approach centralized performance considering its resilience to CSI errors.
A decentralized scheme proposed in \cite{Ref_shi2023decentralized} attains roughly 80\% of the achievable rate of centralized precoding in cell-free URLLC settings. 
Recently, \cite{Femenias2024} argues that improved P-MMSE precoding is sufficiently powerful to eliminate the need for downlink channel estimation, an advantage where its distributed counterpart often {falls} short.

However, { centralized precoders are typically derived under a sum-power constraint}, where the total power budget of the whole system {can be}  arbitrarily allocated among {distributed} access points (APs). {In CF distributed deployments, some APs are very close to a user while others are very distant. Such a high path-loss disparity spanning several orders of magnitude results in a \textit{power concentration} effect, where centralized precoders often concentrate transmission power on a small subset of APs to achieve the desired global spatial directivity. Such “cross-AP” power assignment is unrealistic as a few APs are required to transmit at power levels far exceeding their hardware limits, resulting in {a}  performance drop. }

{Hence, this letter aims to study centralized precoding under realistic instantaneous per-AP power constraint, which is critical for operating within the linear region of power amplifiers (PAs). While prior works \cite{yu2007transmitter, miretti2024duality} established optimal centralized designs under \textit{average} power constraints, extending these to strict instantaneous limits often leads to intractable optimization. Furthermore, iterative solvers applied to small-scale fading (millisecond scale) impose prohibitive computational burdens and processing delays, which can result in overwhelming performance degradation due to channel aging \cite{Ref_jiang2021impactcellfree}. Consequently, this work focuses on standard closed-form linear precoders (e.g., ZF and MMSE) that are widely adopted for their scalability in large-scale deployments.
}

The paper is organized as {follows:} Sec. II introduces the system and channel model. Sec. III formalizes centralized and distributed downlink strategies and derives their achievable spectral efficiency (SE). Sec. IV presents the precoding design and normalization, highlighting how per-AP constraints can degrade centralized performance. Numerical results are discussed in Sec. V, followed by conclusions in Sec. VI.


\section{System Model}

Consider the downlink of a CF-mMIMO system comprising $L$ APs, each equipped with $N_t$ antennas, distributed throughout a coverage area to serve $K$ single-antenna users.  The sets of indices for APs and users are represented by $\mathbb{L}= \{1,\ldots,L\}$ and $\mathbb{K}=\{1,\ldots,K\}$, respectively. Under a user-centric dynamic clustering strategy, user \(k\) is served by a subset of nearby APs, denoted by \(\mathbb{L}_k \subseteq \mathbb{L}\). Consequently, each AP $l$ serves a dynamically assigned subset of users, i.e., \(
\mathbb{K}_l = \{\,k\in\mathbb{K} : l\in\mathbb{L}_k\,\}
\). Leveraging the block fading assumption and time-division duplexing reciprocity, the channel remains constant within each coherence block, which is divided into three phases: uplink training, uplink data transmission (neglected here), and downlink data transmission.

We adopt a general model with spatially-correlated Rician fading for the channel $\mathbf{h}_{kl} \in \mathbb{C}^{N_t}$ between AP $l$ and user $k$. Parameterized by the large-scale attenuation $\beta_{kl}$ and the Rician $K$-factor $\kappa_{kl} \ge 0$, this model is constituted by a deterministic line-of-sight (LoS) contribution and a correlated diffuse component, i.e., 
\begin{equation}
\label{eq:rician-model}
\mathbf{h}_{kl}
= 
\sqrt{ \frac{\beta_{kl}\,\kappa_{kl}}{1+\kappa_{kl}}}\,e^{j\theta_{kl}}\,\mathbf{a}_l( {\vartheta_{kl}} )
+\sqrt{\frac{1}{1+\kappa_{kl}}}\,\,\mathbf{g}_{kl}
,
\end{equation}
where
\begin{itemize}
  \item $\theta_{kl}$ is a random phase shift, typically modeled as uniform distribution: $\theta_{kl} \sim \mathcal{U}[0, 2\pi)${.}
  \item $\mathbf{a}_l({\vartheta_{kl}})\in\mathbb{C}^{N_t}$ is the array steering vector for angle ${\vartheta_{kl}}$.
  \item $\mathbf{g}_{kl}\sim\mathcal{CN}(\mathbf{0},\mathbf{R}_{kl})$ is non-line-of-sight (NLoS) component, where the positive semi-definite covariance matrix $\mathbf{R}_{kl}\succeq\mathbf{0}$ captures the spatial correlation.
\end{itemize}
Because $\mathbf{g}_{kl}$ is Gaussian, the distribution of $\mathbf{h}_{kl}$ (conditioned on $\theta_{kl}, {\vartheta_{kl}}$, and $\mathbf{R}_{kl}$) is complex Gaussian, i.e.,
$\mathbf{h}_{kl}\sim\mathcal{CN}(\boldsymbol{\mu}_{kl},\mathbf{Q}_{kl})$,
with the LoS mean and covariance matrix given by:
\begin{align}
\boldsymbol{\mu}_{kl} &= \sqrt{ \frac{\beta_{kl}\,\kappa_{kl}}{1+\kappa_{kl}}}\,e^{j\theta_{kl}}\,\mathbf{a}_l( {\vartheta_{kl}}  ), \\
\mathbf{Q}_{kl} &= \mathbb{E}\left[ (\mathbf{h}_{kl} - \boldsymbol{\mu}_{kl})(\mathbf{h}_{kl} - \boldsymbol{\mu}_{kl})^H \right] = \frac{1}{\kappa_{kl} + 1} \mathbf{R}_{kl}.
\end{align}

Let $\{\varphi_1,\dots,\varphi_{\tau_p}\}$ be orthogonal pilots with $\|\varphi_t\|^2=\tau_p$. The users simultaneously send their assigned orthogonal pilot during the uplink training phase. 
After correlating the received pilot matrix at AP \(l\) with user \(k\)'s pilot we obtain the sufficient statistic 
\(
\boldsymbol{\phi}_{kl}
= \sum_{i\in\mathcal{P}_k} \sqrt{p_u\tau_p}\,\mathbf{h}_{il} + \mathbf{n}_{l}
\),
where $\mathcal{P}_k$ is the set of users sharing the same pilot as user $k$ (suffering pilot contamination), thermal noise $\mathbf{n}_{l}\sim\mathcal{CN}(\mathbf{0},\sigma^2_n\mathbf{I}_{N_t})$ with variance $\sigma^2_n$, and $p_u$ is the uplink transmission power.
Define
\(
\boldsymbol{\Psi}_{k l} \triangleq p_u\tau_p\sum_{i\in\mathcal{P}_k} \,\mathbf{Q}_{il} + \sigma_n^2\mathbf{I}_{N_t}
\), the conditional (phase-aware) MMSE estimator for user $k$ is expressed by
$
    \hat{\mathbf{h}}_{kl} = \boldsymbol\mu_{kl} + \sqrt{p_u\tau_p}\,\mathbf{Q}_{kl}\,\boldsymbol{\Psi}_{k l}^{-1}
\big(\boldsymbol{\phi}_{kl} - \mathbb{E}[\boldsymbol{\phi}_{kl}]\big)
$.
As a result, the estimate and estimation error $\tilde{\mathbf{h}}_{kl} = \mathbf{h}_{kl} - \hat{\mathbf{h}}_{kl}$ are distributed as:
\begin{align}
\hat{\mathbf{h}}_{kl} &\sim \mathcal{CN}\left( \boldsymbol{\mu}_{kl},\ p_u \tau_p \mathbf{Q}_{kl}\,\boldsymbol{\Psi}_{k l}^{-1}\mathbf{Q}_{kl} \right), \\
\tilde{\mathbf{h}}_{kl} &\sim \mathcal{CN}\left( \mathbf{0},\ \mathbf{Q}_{kl} - p_u \tau_p \mathbf{Q}_{kl}\,\boldsymbol{\Psi}_{k l}^{-1}\mathbf{Q}_{kl} \right) \triangleq \mathcal{CN}\left( \mathbf{0},\ \boldsymbol{\Theta}_{kl} \right).
\end{align}

\section{Downlink Transmission}
\label{sec:transmission_schemes}

In the downlink, the information symbols for the $K$ users are denoted by $\mathbf{s} = [s_1,\ldots,s_K]^T$, whose elements are zero-mean, unit-variance, and mutually uncorrelated, i.e., $\mathbb{E}[\mathbf{s}\mathbf{s}^H] = \mathbf{I}_K$. This section develops signal models and derives the SE expressions for both distributed and centralized strategies. 

\subsection{Distributed Strategy}
\label{subsec:distributed}
In this case, each AP $l$ independently serves its associated users $k\in \mathbb{K}_l$, producing its transmitted signal:
\begin{equation} \label{eq:composite_tx_sig_dist}
    \mathbf{s}_{l} = \sqrt{p_a} \sum\nolimits_{k \in \mathbb{K}_l} \sqrt{\eta_{kl}}\, \mathbf{w}_{kl}  s_k,
\end{equation}
where $p_a$ is the {maximal transmit power per AP, locked by the saturation power of PAs}, $\eta_{kl} \in [0, 1]$ is the power coefficient assigned by AP $l$ for user $k$, satisfying $\sum_{k\in \mathbb{K}_l} \eta_{kl}\leqslant 1$, and $\mathbf{w}_{kl} \in \mathbb{C}^{N_t}$ is the local precoding vector, normalized such that $\mathbb{E}\left[ \|\mathbf{w}_{kl}\|^2 \right] = 1$.  
User $k$ observes $y_k  =   \sum_{l \in \mathbb{L}} \mathbf{h}_{kl}^T \mathbf{s}_l  + n_k$,
where $n_k \sim \mathcal{CN}(0, \sigma_n^2)$. Applying \eqref{eq:composite_tx_sig_dist} yields \begin{equation} \label{eq:downlinkModel}
    y_k \;=\; \sqrt{p_a}\sum\nolimits_{l\in\mathbb{L}}\sum\nolimits_{k'\in\mathbb{K}_l} \sqrt{\eta_{k'l}}\,\mathbf{h}_{kl}^T\mathbf{w}_{k'l}\,s_{k'} \;+\; n_k.
\end{equation} 
\begin{proposition}
\label{prop:se_dist}
An achievable downlink SE for user $k$ under distributed precoding is given by
$R_k = \mathbb{E} \left[ \log_2 \left( 1 + \gamma_k \right) \right]$,
where the instantaneous effective signal-to-interference-plus-noise ratio (SINR) is
\begin{equation} \label{eq:sinr_dist}
    \gamma_k = \frac{ 
        \left|   \sum\nolimits_{l \in \mathbb{L}_k} \sqrt{\eta_{kl}} \mathbb{E} \left[ \mathbf{h}_{kl}^T \mathbf{w}_{kl} \right] \right|^2
    }{ \left\{ \begin{aligned}
         \sum\nolimits_{k' \in \mathbb{K}} &\sum\nolimits_{l \in \mathbb{L}_{k'}} \eta_{k'l}  \mathbb{E} \left[ \left| \mathbf{h}_{kl}^T \mathbf{w}_{k'l} \right|^2 \right] \\
        & - \sum\nolimits_{l \in \mathbb{L}_k} \eta_{kl} \left| \mathbb{E} \left[ \mathbf{h}_{kl}^T \mathbf{w}_{kl} \right] \right|^2 +\dfrac{\sigma_n^2}{p_a}  
    \end{aligned}         \right\}   
    }.
\end{equation}
\end{proposition}
\begin{IEEEproof}
{It follows a standard use-and-forget approach for CF systems; see an analogous derivation in~\cite[Th.~4.6]{SIG-093}. }
\end{IEEEproof}

\subsection{Centralized Strategy}
\label{subsec:centralized}

In contrast to the distributed approach, the CPU leverages the system-wide CSI to jointly compute the precoding vectors for all APs. Let $\mathbf{v}_k \in \mathbb{C}^{M}$ denote the centralized precoding vector for user $k$, normalized such that $\mathbb{E}[||\mathbf{v}_k||^2] = 1$. It can be decomposed into $L$ local segments, i.e., $\mathbf{v}_k = [\mathbf{v}_{k1}^T, \mathbf{v}_{k2}^T, \ldots, \mathbf{v}_{kL}^T]^T$, where $\mathbf{v}_{kl} \in \mathbb{C}^{N_t}$ is the portion applied by AP $l$.  The CPU constructs the composite signal vector for the entire network:
\begin{equation} \label{eq:centralized_tx_sig}
    \mathbf{x} = \sqrt{P_s} \sum \nolimits_{k\in \mathbb{K}} \sqrt{\epsilon_k} \, \mathbf{v}_k s_k,
\end{equation}
where $P_s$ is the total system power, i.e., $P_s=Lp_a$, $\epsilon_k \in [0,1]$ is the power control coefficient assigned to user $k$, satisfying $\sum_{k\in \mathbb{K}} \epsilon_k \leq 1$. 

This full signal vector $\mathbf{x} = [\mathbf{x}_1^T, \mathbf{x}_2^T, \ldots, \mathbf{x}_L^T]^T$ is then distributed to the respective APs via fronthaul links. Each AP $l$ transmits its designated segment $\mathbf{x}_l \in \mathbb{C}^{N_t}$.
The received signal at user $k$ is given by
\begin{equation}
    y_k = \mathbf{h}_k^H \mathbf{x} + n_k = \sqrt{P_s} \, \mathbf{h}_k^H \left( \sum\nolimits_{k'\in \mathbb{K}} \sqrt{\epsilon_{k'}} \, \mathbf{v}_{k'} s_{k'} \right) + n_k,
\end{equation}
where $\mathbf{h}_k=[\mathbf{h}_{k1}^T,\ldots,\mathbf{h}_{kL}^T]^T$  is the collective channel vector from all AP antennas to user $k$, and its estimate is $\hat{\mathbf{h}}_k$.

\begin{proposition}
\label{prop:se_cent}
An achievable downlink SE for user $k$ under centralized precoding is given by $C_k = \mathbb{E} \left[ \log_2 \left( 1 + \xi_k \right) \right]$, where the effective SINR equals
\begin{equation} \label{eq:sinr_cent}
    \xi_k = \frac{
          \epsilon_k  \left| \mathbb{E} \left[ \mathbf{h}_k^T \mathbf{v}_k \right] \right|^2
    }{
         \sum\limits_{k'\in \mathbb{K} } \epsilon_{k'}  \mathbb{E} \left[ \left| \mathbf{h}_k^T \mathbf{v}_{k'} \right|^2 \right] -   \epsilon_k  \left| \mathbb{E} \left[ \mathbf{h}_k^T \mathbf{v}_k \right] \right|^2 +\dfrac{\sigma_n^2}{P_s}
    }.
\end{equation}
\end{proposition}
\begin{IEEEproof}
Follow the same approach as that of Proposition~\ref{prop:se_dist} and is consistent with Eq.~(35) in~\cite{Ref_bjornson2020scalable}.
\end{IEEEproof}

\begin{table*}[ht] 
\caption{Unnormalized Precoders to form spatial directivity}    \label{tab:complexity_summary}
    \centering
    \begin{tabular}{c|c|c|c}
        \hline \hline 
        \textbf{Methods} & CBF/MR &  RZF  & MMSE \\
        \hline
        \text{Full$\big/$Partial}  & $ \hat{\mathds{h}}_k^*$  & $\hat{\mathds{H}}^* (\hat{\mathds{H}}^T \hat{\mathds{H}}^* + \sigma_n^2 \mathbf{I}_M)^{-1}\mathbf{e}_k^M$  &  $\left(\left( p_u \hat{\mathds{H}} \mathbf{E} \hat{\mathds{H}}^H + p_u \sum\nolimits_{k\in \mathbb{K} }  \mathbf{U}_{k}\boldsymbol{\Theta}_{k}\mathbf{U}_{k}  + \sigma_n^2 \mathbf{I}_{M} \right)^{-1} \hat{\mathds{H}}\right)^*\mathbf{e}_k^M$ \\
        \hline
        \text{Local$\big/$Local Partial} & $ \hat{\mathds{h}}_{kl}^*$   & $\hat{\mathds{H}}_l^* (\hat{\mathds{H}}_l^T \hat{\mathds{H}}_l^* + \sigma_n^2 \mathbf{I}_{N_t})^{-1}\mathbf{e}_k^K$  & $ \left(\left( p_u \hat{\mathds{H}}_l \mathbf{E}_{l} \hat{\mathds{H}}_l^H + p_u \sum\nolimits_{k\in \mathbb{K} }  \mathbf{U}_{kl}\boldsymbol{\Theta}_{kl}\mathbf{U}_{kl}  + \sigma_n^2 \mathbf{I}_{N_t} \right)^{-1} \hat{\mathds{H}}_l\right)^* \mathbf{e}_k^K$ \\
        \hline \hline
 \multicolumn{4}{l@{}}{%
            \parbox{\linewidth}{%
                \vspace{0.5ex}
                \raggedright
                \footnotesize\textit{Note:}
                \begin{enumerate}
                    \setlength{\itemsep}{-0.2ex}
                    \setlength{\parskip}{0pt}
                    \setlength{\parsep}{0pt}
                    \raggedright
                    \item $\mathbf{e}_k^{K}$ denotes the $k^{\text{th}}$ column of the $K \times K$ identity matrix $\mathbf{I}_{K}$, and $\mathbf{e}_k^{M}$ denotes the $k^{\text{th}}$ column of $\mathbf{I}_{M}$.
                    \item The RZF precoder reduces to the ZF precoder when the regularization term $\sigma_n^2 \mathbf{I}$ is omitted.
                    \item Power-coefficient diagonal matrix $\mathbf{E} \triangleq \mathrm{diag}(\epsilon_{1}, \epsilon_{2}, \ldots, \epsilon_{K})$, and its local counterpart $\mathbf{E}_l \triangleq \mathrm{diag}(\eta_{1l}, \eta_{2l}, \ldots, \eta_{Kl})$. 
                    \item Error covariance $\boldsymbol{\Theta}_{k} \triangleq \mathrm{blkdiag}\big(\boldsymbol{\Theta}_{k1}, \ldots, \boldsymbol{\Theta}_{kL}\big)$, and the puncturing operation $\mathbf{U}_{k}\boldsymbol{\Theta}_{k}\mathbf{U}_{k}$ zeros the rows/columns of non-serving APs.
                \end{enumerate}
                \vspace{-1.5ex}
            }
        }
    \end{tabular}
\end{table*}

 \section{Precoder Design under Power Constraints}
To better illustrate the impact of the power constraint, we decouple precoding design into two steps: (i) formulating an unnormalized precoder to achieve the desired spatial directivity, and (ii) normalizing it to satisfy the specified power limit.  

\subsection{Unconstrained Precoder Formulation}
Three precoding techniques are typically employed, each targeting a particular spatial directivity. Specifically, CBF a.k.a. maximal-ratio (MR) maximizes the desired signal power toward each user; (regularized) zero-forcing (RZF) nullifies inter-user interference; and MMSE optimally balances interference suppression against noise amplification. Remark that \textit{partial} \cite{Ref_bjornson2020scalable} and \textit{local partial} \cite{Ref_interdonato2020local} precoding represent centralized and distributed implementations, respectively, tailored to user-centric clustering. To represent diverse schemes uniformly, we introduce a \textit{punctured} (effective) channel that reflects the cluster selection.

Begin by defining an $N_t \times N_t$ diagonal selection matrix for each AP-user pair\footnote{We employ an AP-basis selection because all antennas at one AP typically share the same large-scale fading toward a given user. A more general antenna-basis selection is possible (see \cite{Ref_bjornson2020scalable}).}:
\begin{equation}
\mathbf{U}_{kl} \triangleq 
\begin{cases}
\mathbf{I}_{N_t}, & \text{if AP $l$ serves user $k$}, \\
\mathbf{0}_{N_t}, & \text{otherwise}.
\end{cases}
\end{equation}
The punctured local channel estimate from AP $l$ to user $k$ is given by $\hat{\mathds{h}}_{kl} = \mathbf{U}_{kl} \hat{\mathbf{h}}_{kl}$, which equals $\mathbf{h}_{kl}$ if AP $l$ serves user $k$, and the zero vector otherwise.
Then, $\hat{\mathbb{H}}_l = [\hat{\mathds{h}}_{1l}, \hat{\mathds{h}}_{2l},\ldots, \hat{\mathds{h}}_{Kl}] \in\; \mathbb{C}^{N_t\times K}$ denotes the punctured local channel matrix at AP $l$.
Stacking the per-AP selectors for user \(k\) gives the block-diagonal selection matrix $\mathbf{U}_k \triangleq \mathrm{blkdiag}\big(\mathbf{U}_{k1},\mathbf{U}_{k2},\ldots,\mathbf{U}_{kL}\big)$, and yields the punctured {global} channel matrix $\hat{\mathds{H}} \triangleq \big[\hat{\mathds{h}}_1,\;\hat{\mathds{h}}_2,\;\ldots,\;\hat{\mathds{h}}_K\big]\in\mathbb{C}^{M\times K} $, where $\hat{\mathds{h}}_k = \mathbf{U}_k \hat{\mathbf{h}}_k$. 
Using this notation, different precoding schemes (cell-free full and user-centric partial), for both centralized and local implementations, are denoted uniformly, as we summarized in Table \ref{tab:complexity_summary}.

\subsection{Power-Constrained Normalization}
The above unnormalized precoders specify only the desired spatial directivity; they must be subsequently normalized to ensure compliance with the  power constraint.
\paragraph{ Distributed Per-AP Normalization}

Let $\mathbf{d}_{kl} \in \mathbb{C}^{N_t}$ denote an arbitrary local precoder, as specified in the second row of the table.  The normalized precoder is then obtained as
\begin{equation} \label{GS_localprecoder}
    \mathbf{w}_{kl} = \frac{\mathbf{d}_{kl}}{c_{kl}}, \quad \forall l \in \mathbb{L}, \; \forall k \in \mathbb{K},
\end{equation}
where $c_{kl}$ is a normalization constant. In the literature, two normalization schemes  \cite{khoshnevisan2012power} are commonly used: 1)  
\textit{Long-term normalization} sets $c_{kl} = \sqrt{\mathbb{E}\left[\|\mathbf{d}_{kl}\|^2\right]}$, ensuring $\mathbb{E}[\|\mathbf{w}_{kl}\|^2]=1$ on average \cite{Ref_interdonato2020local}; and 2) \textit{Short-term normalization} uses $c_{kl}=\|\mathbf{d}_{kl}\|$, giving $\|\mathbf{w}_{kl}\|^2=1$ in every channel realization \cite{femenias2020shortterm }.  

From a practical standpoint, we adopt short-term normalization, which is critical for compliance with:
\begin{enumerate}
    \item \textit{Instantaneous Power Limits:} Regulatory standards typically impose a strict maximum transmit power that must not be exceeded at any time instant on each AP.
    \item \textit{Amplifier Saturation:} Power amplifiers in each AP have a finite saturation point; the peak power of the transmitted signal must not exceed this maximum output.
\end{enumerate}
Substituting $\mathbf{w}_{kl} = \frac{\mathbf{d}_{kl}}{\|\mathbf{d}_{kl}\|}$ into  \eqref{eq:composite_tx_sig_dist} obtains the instantaneous transmit power at AP $l$ as $\|\mathbf{s}_{l}\|^2 = p_a \sum_{k \in \mathbb{K}_l} \eta_{kl}$\footnote{ { For simplicity, we assume constant-envelope modulation (e.g., QPSK and 8PSK) where $|s_k|^2 = 1$ holds deterministically. Under general modulation (e.g., QAM), the symbol power varies, and the expressions should be interpreted in expectation over the symbols, i.e., $\mathbb{E}[|s_k|^2] = 1$. Critically, the symbol-induced power variance (e.g., $\approx 0.32$ for 16-QAM and $0.38$ for 64-QAM) remains substantially smaller than the power-concentration effects, the primary focus of our analysis.}  }. 
Since $\sum_{k \in \mathbb{K}_l} \eta_{kl} \leqslant 1$, it follows that $\|\mathbf{s}_{l}\|^2 \leqslant p_a$, which conforms with each AP's power limit $p_a$.

\paragraph{Centralized Sum-Power Normalization}
Similarly, let $\mathbf{d}_{k} \in \mathbb{C}^{M}$ denote the unnormalized centralized precoding vector for user $k$, as specified in the first row of the table. The normalized precoder is obtained by
\begin{equation} \label{sumPower_normalization}
    \mathbf{v}_{k} = \frac{\mathbf{d}_{k}}{c_k}, \quad  \forall k \in \mathbb{K},
\end{equation}
where $c_k$ can be defined by either a long-term ($c_k=\sqrt{\mathbb{E}\left[\|\mathbf{d}_{k}\|^2\right]}$) or short-term ($c_k=\|\mathbf{d}_{k}\|$) normalization factor. Substituting \eqref{sumPower_normalization} into \eqref{eq:centralized_tx_sig} shows the sum power across all APs is constrained, i.e., $\mathbb{E}[\|\mathbf{x}\|^2] = \sum_{k \in \mathbb{K}} \epsilon_{k} \mathbb{E}[\|\mathbf{x}_l\|^2] \leqslant P_s$.
\begin{remark} \textit{
Prior works commonly adopt a long-term sum-power constraint — see, e.g., Eq.(36) of \cite{Ref_bjornson2020scalable}, \cite{miretti2021team}, Eq. (15) of \cite{Femenias2024},  and \cite{miretti2021precoding}. Some other works implicitly enforced it, e.g., Eq.~(11) in~\cite{Ref_nayebi2017precoding} and Eq.~(19) in~\cite{du2021cell} mandate that normalization factors ($\lambda_{l,k}$) and power-control coefficients ($\eta_{kl}$), respectively, be identical for any user $k$ across all APs (i.e., $\lambda_{l,k} = \lambda_k$ and $\eta_{kl} = \eta_k$, $\forall l\in\mathbb{L}$). Although not explicitly stated, it prohibits per-AP normalization and is functionally equivalent to the sum-power power constraint described by \eqref{sumPower_normalization}. }
\end{remark}

\subsection{Per-AP Constrained Centralized Precoding}
In the centralized approach, AP~$l$ transmits 
\begin{equation}
    \mathbf{x}_l = \sqrt{P_s} \sum\nolimits_{k\in \mathbb{K}} \sqrt{\epsilon_k} \, \mathbf{v}_{kl} s_k
\end{equation} with the power emitted by AP~$l$ 
\begin{equation} \label{GS_TxPowerl}
    P_l = P_s \sum_{k=1}^{K} \epsilon_k \, \|\mathbf{v}_{kl}\|^2
\end{equation} 
There is \emph{no guarantee} that $P_l\le p_a$ for every AP $l$. This kind of “cross-AP” power assignment is unrealistic, since some APs might be driven to transmit above their hardware limits, as illustrated in \figurename \ref{fig:powerdist1}. { As PA saturation imposes an instantaneous hardware limit that cannot be addressed by time- or frequency-domain averaging,}  we need {to} restrict either the precoders $\{\mathbf{v}_k\}$ or power coefficients $\{\epsilon_k\}$ {to enforce per-AP instantaneous power constraint}. In the sequel we present two {heuristic} approaches\footnote{ {Existing optimal centralized design \cite{miretti2024duality} is tailored for average per-AP power constraints; extending it to stricter instantaneous limits yields intractable optimization. As this iterative solver must execute for every small-scale fading realization (on a millisecond scale), imposing prohibitive computational burden. Moreover, the centralized process incurs significant delay from global CSI gathering, iterative solving, and precoded data distribution. This delay can lead to channel aging, thereby diminishing the theoretical gains \cite{Ref_jiang2021impactcellfree}. These limitations necessitate the low-complexity, real-time heuristic normalizations adopted here, while leaving the optimal design under instantaneous per-AP limit for future works.} } to enforce this limit:
\begin{figure}[!t] 
    \centering
    \includegraphics[width=0.37\textwidth]{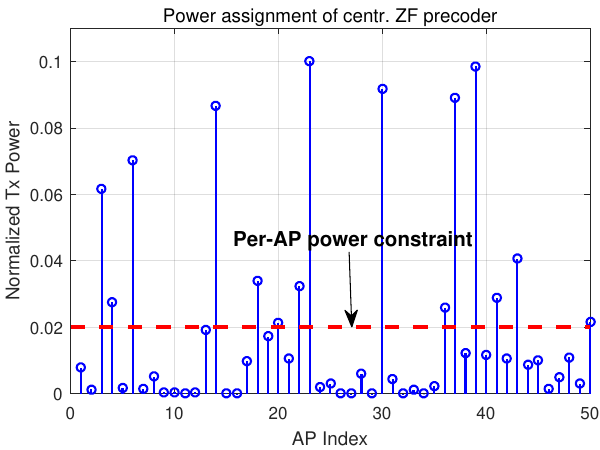}
     \caption{ {Illustration of the power concentration effect:}  distribution of normalized transmit power across APs for a CF system with $50$ single-antenna APs and $10$ users, using centralized ZF precoder with suboptimal power allocation of $\epsilon_1=\ldots=\epsilon_K$ (following Eq. (21) in \cite{Ref_nayebi2017precoding}). The horizontal (red) dashed line marks the normalized per-AP power constraint of $1/L = 0.02$.  }
     \label{fig:powerdist1}    
\end{figure}

\paragraph{Per-AP Constrained Power Scaling (PS)}
Let $P_{\max} \triangleq \max_{l\in\mathbb{L}} P_l$.
If \(P_{\max}> p_a\), scale down all power coefficients by  $ \bar{\epsilon}_k \;=\; \alpha_g\,\epsilon_k$
$\forall k$, using a factor of $\alpha_g \;=\; \frac{p_a}{P_{\max}} \in (0,1) $\footnote{ The condition $\alpha_g < 1$ indicates that unscaled precoding weights violate the per-AP instantaneous power limit $p_a$. In distributed CF topologies, the channel matrix $\mathbf{H}$ contains coefficients spanning several orders of magnitude because some APs are very close to a user while others are very distant. Centralized precoders inherently exploit this disparity by concentrating power on a few ``bottleneck'' APs to maximize spatial directivity, resulting in peak demands $P_{\max} \gg p_a$ and a severe global back-off $\alpha_g \ll 1$. {While a closed-form statistical characterization is analytically challenging, the power concentration effect is clearly explained and empirically validated.} }.
Under this update every AP meets the power constraint, as
\begin{equation}
\bar{P}_l \;=\; P_s\sum\nolimits_{k\in \mathbb{K} } \bar{\epsilon}_k \|\mathbf{v}_{kl}\|^2
             \;=\; \alpha_g\, P_l \;\le\; p_a.
\end{equation}
\textit{Drawbacks:} While this method preserves the global spatial directivity of the centralized precoders, i.e., $\|\mathbf{v}_k\|^2 = 1$, it can significantly reduce the total transmission power by $\sum_{k \in \mathbb{K}} \bar{\epsilon}_k \ll 1$, when $P_{\max} \gg p_a$.


\paragraph{Per-AP Constrained Precoder Local Normalization (LN)}
Each local portion of the centralized precoder is independently normalized as 
\begin{equation}
    \bar{\mathbf{v}}_{kl} = \frac{\mathbf{v}_{kl}}{\sqrt{L} \, \| \mathbf{v}_{kl}\| }, \quad  \forall k \in \mathbb{K},
\end{equation} where the use of the factor $\frac{1}{\sqrt{L}}$ is to guarantee that the updated  precoder 
$\bar{\mathbf{v}}_k = [\bar{\mathbf{v}}_{k1}^T, \bar{\mathbf{v}}_{k2}^T, \ldots, \bar{\mathbf{v}}_{kL}^T]^T$ still satisfies 
$\|\bar{\mathbf{v}}_k\|^2 = 1$. 
Replacing $\mathbf{v}_{kl}$ in \eqref{GS_TxPowerl} with $\bar{\mathbf{v}}_{kl}$, the transmit power at AP $l$ becomes
\begin{equation}
\bar{P}_l = P_s \sum\nolimits_{k\in \mathbb{K}} \epsilon_k  \|\bar{\mathbf{v}}_{kl}\|^2 = \frac{P_s}{L} \sum\nolimits_{k\in \mathbb{K}} \epsilon_k \leqslant p_a,
\label{eq:perAP_power}
\end{equation}
which ensures that per-AP power constraint {is met}.  
\textit{Drawbacks:} This local normalization distorts the intended spatial directivity, thereby degrading the performance of the centralized precoder.

\begin{figure*}[!tbph]
\centerline{
\subfloat[]{
\includegraphics[width=0.33\textwidth]{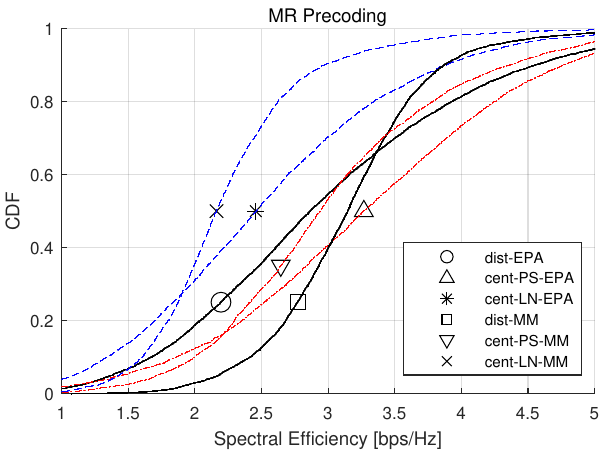}
\label{fig:mr}
}
\hspace{0mm}
\subfloat[]{
\includegraphics[width=0.33\textwidth]{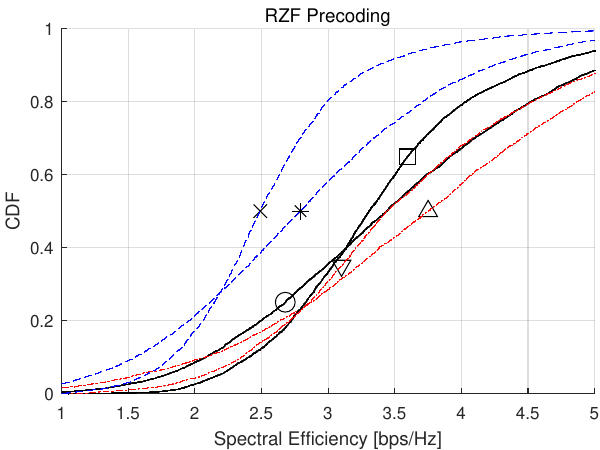}
\label{fig:rzf}
}
\hspace{0mm}
\subfloat[]{
\includegraphics[width=0.33\textwidth]{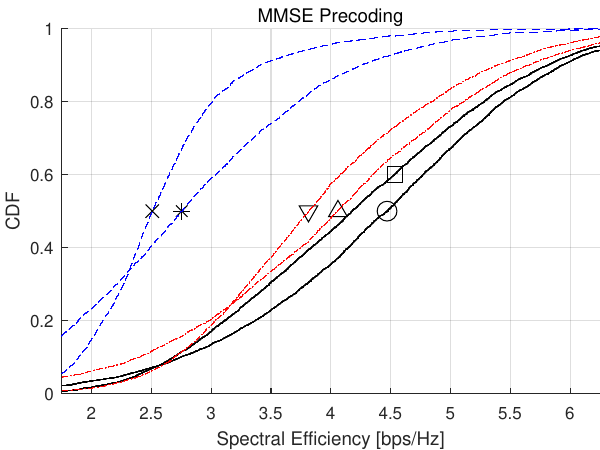}
\label{fig:mmse}
}
}
\hspace{15mm}
\caption{Performance comparison between distributed and centralized precoding for (a) MR, (b) RZF, and (c) MMSE, under both equal-power and max–min power optimization. All three subfigures use the same markers as in (a) to denote the different options. }
\label{Fig_performance}
\end{figure*}

\section{Simulation Results and Discussions}
The simulation is set up using the following parameters:
\begin{itemize}
    \item Number of distributed APs: \( L = 50 \), each with \( { N_t} = 4 \) antennas (total antennas \( M = 200 \))
    \item Number of active users: \( K = 10 \)
    \item User-centric clustering: a user served by {10 nearest APs}
    \item Coverage radius: 1 km
    \item Maximum power at user terminal: \( p_u = 200 \) mW
    \item Maximum power at AP: \( p_a = 200 \) mW
    \item Noise power spectral density: --174 dBm/Hz
    \item Noise figure: 9 dB; \: Bandwidth: 5 MHz
    \item Antenna array: Uniform linear arrays with $\lambda/2$ spacing
    \item Large-scale fading: COST-Hata model (as in \cite{Ref_ngo2017cellfree})
    \item Rician K-factor: $\kappa_{kl} \sim \mathcal{N}(8, 4^2)$ (see 3GPP TR 38.901)
    \item Spatial correlation: Gaussian scattering~\cite[Sec.~2.6]{SIG-093} with angular spread of $10^\circ$
    \item Coherence block length: $\tau_c = 200$ channel uses
    \item Pilot length: $\tau_p = 5$, resulting in pilot contamination
\end{itemize}

The cumulative distribution functions (CDFs) of spectral efficiency for the three precoding methods are illustrated in { \figurename \ref{fig:mr} to \ref{fig:mmse} }, comparing distributed (dist) with centralized (cent) precoding under equal power allocation (EPA) and max-min  (MM) power control. The $5^{th}$ percentile point on each CDF indicates the $95\%$-likely SE, the most important metric used to quantify uniform service and user fairness.  First, \figurename \ref{fig:mr} presents the CDF for MR precoding. The distributed variants (dist-EPA and dist-MM) show robust performance across the SE distribution, with max-min power control effectively improving fairness by enhancing the lower percentiles without significant sacrifice in higher ones. In contrast, centralized methods exhibit shifts to lower SE values, particularly evident in cent-LN-EPA and cent-LN-MM, where the enforcement of per-AP power constraints causes notable SE degradation. The cent-PS-MM curve achieves strong performance at the upper percentiles but still falls short of distributed MR precoding in fairness.

Since each {AP} may have fewer antennas than users (${ N_t} < K$), making $\mathbf{H}_l$ rank-deficient and ZF infeasible, RZF is adopted in our simulations. The regularization term ensures matrix invertibility and enables reliable precoding. As shown in \figurename~\ref{fig:rzf}, similar trends are observed: distributed precoding consistently achieves higher SE probabilities, particularly under max-min (MM) power control, which enhances the $95\%$-likely SE while maintaining overall efficiency. Centralized RZF with LN suffers the most, with its CDF curve shifted left, indicating reduced SE for the majority of users due to aggressive per-AP normalization. PS alleviates this degradation somewhat but still fails to outperform distributed schemes.

In contrast to Fig. 6 in \cite{Ref_bjornson2020scalable}, where centralized precoding (MMSE and P-MMSE) substantially outperforms LP-MMSE, \figurename~\ref{fig:mmse} demonstrates that distributed MMSE precoding dominates, with dist-MM achieving excellent fairness ($95\%$-likely SE of 2.38 bps/Hz). Among centralized schemes, cent-PS-MM attains a comparable $95\%$-likely SE of 2.42 bps/Hz but at the cost of reduced performance for high-percentile users. The LN-based centralized variants perform poorly, with $95\%$-likely SE values of only 1.23 bps/Hz for equal power allocation and 1.72 bps/Hz for max-min power control. {We further compare their performance in a high antenna-to-user ratio scenario ($N_t=8, K=4$), where centralized precoding still under-performs. In a nutshell, distributed precoding exhibits high robustness under realistic hardware limits}.

\section{Conclusion}
{
This letter revealed that closed-form centralized precoding, designed under sum-power constraints, suffers from a severe power concentration effect. Due to the high path-loss disparity in cell-free topologies, specific APs exhibit disproportionate power demands to achieve global spatial directivity, frequently violating realistic per-AP hardware limits. Enforcement of global scaling via low-complexity, real-time normalization methods degrade significantly centralized performance. Our results confirm that distributed precoding is a robust strategy under realistic hardware limits, in addition to its inherent merits of lower fronthaul overhead and processing latency. {Future work will explore low-complexity per-AP-constraint-aware designs and wideband extensions.}
}

\bibliographystyle{IEEEtran}
\bibliography{IEEEabrv,Ref_COML}

\end{document}